\documentclass[11pt]{article}
\usepackage{fullpage}

\usepackage{amsthm}
\usepackage{appendix}
\usepackage{hyperref}
\usepackage{amsmath}
\usepackage{amssymb}
\usepackage{graphicx}
\usepackage{bm}
\usepackage{amsfonts}
\usepackage{latexsym}
\usepackage{pdfsync}
\usepackage{subfigure}
\usepackage{colordvi}
\usepackage{color}
\usepackage{appendix}

 \usepackage{multirow}
  \usepackage{tabularx}
 \usepackage[english]{babel}
 \usepackage{array}

\newcommand{\commentout}[1]{}

\newcommand{\nwc}{\newcommand}
\nwc{\ba}{\begin{array}}
\nwc{\bal}{\begin{align}}
\nwc{\bea}{\begin{eqnarray}}
\nwc{\beq}{\begin{eqnarray}}
\nwc{\bean}{\begin{eqnarray*}}
\nwc{\beqn}{\begin{eqnarray*}}
\nwc{\beqast}{\begin{eqnarray*}}

\nwc{\ea}{\end{array}}
\nwc{\eal}{\end{align}}
\nwc{\eea}{\end{eqnarray}}
\nwc{\eeq}{\end{eqnarray}}
\nwc{\eean}{\end{eqnarray*}}
\nwc{\eeqn}{\end{eqnarray*}}
\nwc{\eeqast}{\end{eqnarray*}}

\nwc{\ep}{\varepsilon}
\nwc{\ept}{\epsilon}

\nwc{\nwt}{\newtheorem}

\nwt{proposition}{Proposition}
\nwt{corollary}{Corollary}
\nwt{theorem}{Theorem}
\nwt{summary}{Summary}
\nwt{lemma}{Lemma}
\nwt{remark}{Remark}
\nwc{\nn}{\nonumber}

\newcommand{\NN}{\mathbb{N}}

\newcommand{\CC}{\mathbb{C}}
\newcommand{\RR}{\mathbb{R}}
\newcommand{\ZZ}{\mathbb{Z}}
\newcommand{\TT}{\mathbb{T}}
\newcommand{\om}{\omega}

\newcommand{\rank}{\text{Rank\,}}
\newcommand{\supp}{\mathcal{S}}
\nwc{\calP}{\mathcal{P}}
\nwc{\calL}{\mathcal{L}}

\nwc{\xmax}{x_{\text{max}}}
\nwc{\xmin}{x_{\text{min}}}
\nwc{\yep}{y^\ep}
\nwc{\range}{{\rm Range}}
\nwc{\calE}{\mathcal{E}}
\nwc{\calN}{\mathcal{N}}
\nwc{\smax}{\sigma_{\rm max}}
\nwc{\smin}{\sigma_{\rm min}}
\nwc{\re}{{\rm Re}}
\nwc{\im}{{\rm Im}}
\nwc{\bv}{{\rm v}}
\nwc{\bw}{{\rm w}}
\nwc{\PO}{\calP_1}
\nwc{\PT}{\calP_2}
\nwc{\cP}{\mathcal{P}}
\nwc{\PEO}{\calP^\ep_1}
\nwc{\PET}{\calP^\ep_2}
\nwc{\JE}{J^\ep}
\nwc{\RE}{R^\ep}
\nwc{\bc}{{\rm c}}
\nwc{\si}{\sigma}
\nwc{\se}{\sigma^\ep}
\nwc{\UE}{U^\ep}
\nwc{\VE}{V^\ep}
\nwc{\SE}{\Sigma^\ep}
\nwc{\Phit}{ \Phi^{-\frac L 2\rightarrow \frac L 2}}
\nwc{\lan}{\langle}
\nwc{\ran}{\rangle}

\nwc{\HH}{\mathcal{H}}
\nwc{\HHE}{\mathcal{H}^\ep}
\nwc{\EE}{\mathcal{E}}
\nwc{\SI}{\Sigma}
\nwc{\QE}{Q^\ep}
\nwc{\cS}{\mathcal{S}}

\nwc{\XP}{X(\Phi^{M-L})^T}
\nwc{\plo}{\phi^L(\om)}

\title{Compressive  Spectral Estimation with Single-Snapshot ESPRIT: Stability and Resolution}

\author{Albert Fannjiang
\thanks{Department of Mathematics, University of California, Davis, CA. Email:  fannjiang@math.ucdavis.edu. }
}

\begin{document}

\maketitle

\begin{abstract}
In this paper Estimation of Signal Parameters
via Rotational Invariance Techniques (ESPRIT) is developed for  spectral estimation with single-snapshot measurement.  Stability and resolution analysis with performance guarantee for Single-Snapshot ESPRIT (SS-ESPRIT) is the main focus. 

In the noise-free case {\em exact} reconstruction is guaranteed for any arbitrary set of  frequencies as long as the number of measurement data is at least twice the number of
distinct frequencies to be recovered. In the presence of noise and under the assumption that the true frequencies are separated by at least two times Rayleigh's Resolution  Length, an
explicit error bound for frequency reconstruction is given in terms of the dynamic range and the separation of
the frequencies. The separation and sparsity constraint compares favorably with those
of 
the leading approaches to compressed sensing in the continuum. 
\end{abstract}


\section{Introduction}

Suppose a signal $y(t)$ consists of linear combinations of $s$ Fourier components from the set 
$$\{e^{-2\pi i \om_j t} : \om_j \in \RR, \ j = 1,\ldots,s\}.$$
Suppose the external noise $\ep(t)$ is present in the received signal
\beq
\label{model}
y^\ep(t) = y(t) + \ep(t), \quad y(t) = \sum_{j=1}^s x_j e^{-2\pi i \om_j t}. 
\eeq

The problem of spectral estimation is to  recover the frequency support set $\supp =	\{\om_1, . . . , \om_s\}$ 	and	the	corresponding	amplitudes	$x	=	[x_1, . . . , x_s]^T$	from	a	finite data sampled at, say, $t=0,1,2, \cdots, M\in \NN$. Because of the nonlinear dependence of  the signal $y(t)$ 
on frequency, the main difficulty  of spectral estimation lies in identifying $\supp$. The amplitudes $x$ can be easily recovered by solving least squares once $\supp$ is known.

Denote (with a slight abuse of notation)  $y = [y_k]_{k=0}^{M},$ $\ep = [\ep_k]_{k=0}^{M}$
and $y^\ep = y + \ep \in \CC^{M+1}$,  with $y_k = y(k)$, $\yep_k = \yep(k)$ and $\ep_k = \ep(k)$.
Let 
\beq
\label{imagingvector}
\phi^{M}(\om) = [1  \ e^{-2\pi i \om} \ e^{-2\pi i 2\om} \ \ldots \ \ e^{-2\pi i M\om}]^T \in \CC^{M+1}
\eeq 
be  the imaging vector of size $M+1$ at the frequency $\om$ and
define 
$$\Phi^{M}= [\phi^{M}(\om_1)\  \phi^{M}(\om_2) \ \ldots \ \phi^{M}(\om_s)] \in \CC^{(M+1)\times s }.$$ 
The single-snapshot  formulation of spectral estimation takes the form 
\beq
y^\ep = \Phi^{M}  x + \ep. 
\label{linearsystem}
\eeq
In addition to  the nonlinear dependence of $\Phi^M$ on the unknown frequencies,  with the sampling times $t=0,1,2, \cdots, M\in \NN$, one can only hope to determine frequencies
on  the torus $\TT=[0,1)$ with the natural metric 
\[
d(\om_j,\om_l) = \min_{n\in \ZZ} |\om_j+n-\om_l|.
\]

A key unit of frequency separation  is Rayleigh's Resolution  Length (RL),  the distance between the center and the first zero of the sinc function ${\sin{(\pi \om M)}}/{(\pi \om)},$
namely, 1 RL $= 1/M$. 


\subsection{Single-Snapshot ESPRIT (SS-ESPRIT)}
In this paper, to circumvent the gridding problem, we  reformulate the spectral estimation problem (\ref{linearsystem})  in the form of {\em multiple measurement vectors} suitable for the application of Estimation of Signal Parameters
via Rotational Invariance Techniques (ESPRIT)
 \cite{PRK,RK89}. 
\commentout{The MUSIC algorithm was introduced by Schmidt \cite{Sch,SchD} and many extensions  exist including S-MUSIC\cite{SMUSIC}, IES-MUSIC\cite{IESMUSIC}, R-MUSIC\cite{RMUSIC} and RAP-MUSIC\cite{RAPMUSIC}. According to Wikipedia, in a detailed evaluation based on thousands of simulations, M.I.T.'s Lincoln Laboratory concluded that, among currently accepted high-resolution algorithms, MUSIC was the most promising and a leading candidate for further study and actual hardware implementation.} 

Most state-of-the-art spectral estimation methods (\cite{PT13,BRD08, SAS} and references therein)
 assume many snapshots of array measurement as well as 
 statistical assumptions on  measurement noise. Below
 we present a stability and resolution analysis for a {\em deterministic, single-snapshot} formulation of ESPRIT.

Fixing a positive integer $1\leq L < M$, we form the Hankel matrix 
\beq
\label{hankel}
H = {\rm Hankel}(y) =
\begin{bmatrix}
y_0 & y_1 & \ldots & y_{M-L}\\
y_1 & y_2 & \ldots & y_{M-L+1}\\
\vdots & \vdots & \vdots & \vdots\\
y_{L} & y_{L+1} & \ldots & y_{M}\\
\end{bmatrix}.
\eeq
 
It is straightforward to verify that  ${\rm Hankel}(y)$ with $y=\Phi^M x$  admits  the Vandermonde decomposition  
\beq
\label{van}
H = \Phi^L X (\Phi^{M-L})^T, \quad X = {\rm diag}(x_1,\ldots,x_s)
\eeq
with the Vandermonde matrix
$$\Phi^L = \begin{bmatrix}
1 & 1 & \ldots & 1 \\
e^{-2\pi i\om_1} & e^{-2\pi i\om_2}& \ldots & e^{-2\pi i\om_s}\\
(e^{-2\pi i\om_1})^2 & (e^{-2\pi i\om_2})^2& \ldots & (e^{-2\pi i\om_s})^2\\
\vdots &\vdots& \vdots & \vdots\\
(e^{-2\pi i\om_1})^{L} & (e^{-2\pi i\om_2})^{L}& \ldots & (e^{-2\pi i\om_s})^{L}\\
\end{bmatrix}.$$

Let $H_1$ and $H_2$ be two sub-matrices  of $H$ consisting, respectively,  of the first 
and last $L$ rows of $H$. Clearly we have as before \beq
 \label{H_1} 
H_1 &= &\Phi^{L-1} X (\Phi^{M-L})^T,\label{H1} \\
H_2 &=& \Phi^{L-1} \Lambda X(\Phi^{M-L})^T,\quad \Lambda = \text{diag}(e^{-2 \pi i\omega_1}, \dots, e^{-2 \pi i \omega_s}) 
\eeq
which can be rewritten as
\beq
\label{10}
H_1=\Phi^{L-1}Y,\quad H_2=\Phi^{L-1}\Lambda Y, \\
 Y\equiv X (\Phi^{M-L})^T\in \CC^{s\times (M-L+1)}. \label{11}
\eeq
Since $Y$ has full (row) rank, $YY^\dagger=I$ where $Y^\dagger$ denotes the pseudo-inverse of $Y$. Hence from (\ref{10}) we have
\beq
H_2=H_1\Psi \label{12}
\eeq
with $\Psi= Y^\dagger \Lambda Y$
implying  that $\{e^{-i2 \pi \omega_1}, \dots, e^{-i2 \pi \omega_s}\}$  is  the set of nonzero eigenvalues of the unknown $(M-L+1)\times (M-L+1)$ matrix $\Psi$.

\begin{theorem}
\label{thm0} For  the Hankel matrices $H_1, H_2$ given above, 
\[
\Psi=H_1^\dagger H_2
\]
is a rank-$s$ solution to
eq. (\ref{12}).
\end{theorem}
\begin{proof} Since $H_1H_1^\dagger$ is the identity map on the range of $H_1$, it suffices to prove 
$\range (H_1)=\range (H_2)=\range (\Phi^{L-1})$ which 
would follow from $\rank(\Phi^{M-L}) = s$. 

On the other hand, we have
$\rank(\Phi^{L-1}) = s$ if $L \ge s$ and $\om_k \neq \om_l, \ \forall k \neq l$.
This is because  $s\times s$ square submatrix $\Phi_s$ of $\Phi^L$ is a square Vandermonde matrix whose
determinant is given by
\[
\det(\Psi_s)=\prod_{1\leq i<j\leq s}(e^{-i2\pi \om_j} -e^{-i2\pi \om_i}).
\]
Clearly, $\Phi_s$  is invertible   if and only if $\om_i\neq \om_j, i\neq j$. 
Hence $\rank (\Phi_s)=s$ which implies $\rank(\Phi^{L-1}) = s$. 

\end{proof}

SS-ESPRIT is based on the following observation.
\begin{theorem} For the Hankel matrices $H_1$ and $H_2$ given above, let $\Psi$ be
any rank-$s$ solution to $H_2=H_1\Psi.$
\commentout{ is unique and given by
\beq
\label{18}
\Psi=\Big(X (\Phi^{M-L})^T\Big)^\dagger \Lambda X (\Phi^{M-L})^T. 
\eeq 
}
Then 
$\{e^{-i2 \pi \omega_1}, \dots, e^{-i2 \pi \omega_s}\}$ is the set of nonzero eigenvalues of  $\Psi$. 
\label{thm1}
\end{theorem}
\begin{remark}
Theorem \ref{thm1} implies  that the number of measurement data $(M+1) \ge 2s$ 
suffices to guarantee exact reconstruction. \end{remark}

\begin{proof}
From (\ref{10}) we have
\[
\Phi^{L-1}\Lambda Y=\Phi^{L-1} Y\Psi
\]
and hence
\beq
\label{13}
\Lambda Y=Y\Psi
\eeq
since $\Phi^{L-1}$ has full column rank.
Using (\ref{11}) and  transposing (\ref{13}) we obtain
\beqn
\Phi^{M-L}X\Lambda=\Psi^T\Phi^{M-L}X
\eeqn
implying
\beq\label{14}
\Phi^{M-L}\Lambda=\Psi^T\Phi^{M-L}
\eeq
since $X$ is diagonal, full rank and commutes with $\Lambda$.
Eq. (\ref{14}) means that the columns of $\Phi^{M-L}$ are the
eigenvectors of  the matrix $\Psi^T$ with the diagonal entries of $\Lambda$ as the corresponding eigenvalues. 

\end{proof}

Theorems \ref{thm0} and  \ref{thm1} motivate the following reconstruction procedure in the case
of noisy data.

Let $H^{\ep} = {\rm Hankel}(\yep)=H+E $ where $E = {\rm Hankel}(\ep)$.
Extracting $H_1^\ep$ and $H_2^\ep$ analogously from $H^\ep$ we have
\[
H_1^\ep=H_1+E_1,\quad H_2^\ep=H_2+E_2
\]
where $E_1$ and $E_2$ are two sub-matrices  of $E$ consisting, respectively,  of the first 
and last $L$ rows of $E$. 

\commentout{Let  the Singular Value Decomposition (SVD) of $H_1$ be written as
$$ H_1 = [\underbrace{U_1}_{L \times s} \ \underbrace{U_2}_{L\times (L-s)}] \ \underbrace{\text{diag}(\sigma_1,\sigma_2,\ldots,\sigma_s,0,\ldots,0)}_{L \times (M-L+1)} \ [\underbrace{V_1}_{(M-L+1) \times s} \ \underbrace{V_2}_{(M-L+1) \times (M-L+1-s)}]^\star$$
with the singular values $\sigma_1\geq \sigma_2\geq \sigma_3 \geq \cdots\sigma_s>0.$ }Let  the SVD  of $H_1^\ep$ be written as   
\[
H^\ep_1 = [\underbrace{U^\ep_1}_{L \times \tau} \ \underbrace{U^\ep_2}_{L\times (L-\tau)}] \ \underbrace{\text{diag}(\sigma^\ep_1,\sigma^\ep_2,\ldots,\sigma^\ep_s,\sigma^\ep_{s+1},\ldots)}_{L\times (M-L+1)} \ [\underbrace{V^\ep_1}_{(M-L+1) \times \tau} \ \underbrace{V^\ep_2}_{(M-L+1) \times (M-L+1-\tau)}]^\star
\]
with the singular values $\sigma^\ep_1\geq \sigma^\ep_2\geq \sigma^\ep_3 \geq \cdots\ge \sigma_L^\ep.$
Let  $\sigma_1\geq \sigma_2\geq \sigma_3 \geq \cdots\sigma_s>0$ be the nonzero singular values of $H_1$. 
Without loss of generality, we assume $L\leq M-L+1$ or equivalently $2L\leq M+1$.

The number of  frequencies $s$ may be estimated when there is a significant spectral gap.  For instance, 
according to \cite{Adam}, the spectral norm $\|E_1\|_2 $ of a random Hankel
 matrix from a zero mean,  independently and identically distributed (i.i.d.)  sequence of a finite variance is on the order of $\sqrt{M \log M}$ for $M\gg 1$ while 
 $\sigma_s$ for well-separated frequencies is $\mathcal{O}(M)$ (see next section). 
 Hence by  Weyl's theorem \cite{SS}
\beq
\label{Weyl}
|\si^{\ep}_j -\si_j| \le \|E_1\|_2, \quad  j = 1,2,\ldots,L
\eeq
the sparsity $s$ can be easily estimated based on the singular value distribution of $H^\ep$. 
Indeed, a spectral gap emerges because
$\si^\ep_j \le \|E_1\|_2, \ \forall  j\ge s+1$ and $\si^\ep_s \ge \si_s-\|E_1\|_2$.

Suppose the sparsity $s$ is known and set $\tau=s$. Let $\cP_s=U_1^\ep(U_1^\ep)^\star$ denotes the orthogonal projection onto the
singular subspace of the  $s$ largest singular values. Consider the equation
\beq
\label{15}
\cP_s  H_2^\ep=\cP_s H_1^\ep\Psi^\ep,
\eeq
equivalent to
\beq
U_1^\ep \Sigma^\ep_s(V_1^\ep)^\star\Psi^\ep
=\cP_s H_2^\ep,\quad \Sigma^\ep_s= \text{diag}(\sigma^\ep_1,\sigma^\ep_2,\ldots,\sigma^\ep_s) \label{16}
\eeq
Eq. (\ref{16})  can then be solved as
\beq
\hat \Psi
&=&\hat H_1^\dagger H_2^\ep\label{17}\\
\hat H_1&=&U_1^\ep \Sigma^\ep_s(V_1^\ep)^\star=\cP_s H_1^\ep \label{18}
\eeq
with  rank-$s$ $\hat \Psi$. 
Eq. (\ref{17})-(\ref{18}) defines the main steps of Single-Snapshot ESPRIT (SS-ESPRIT). The rest is to find
the nonzero eigenvalues of $\Psi^\ep$ and retrieve the frequencies
from these eigenvalues.

\commentout{
Let $\xmax = \max_{j=1,\ldots,s} |x_j|$ and $\xmin = \min_{j=1,\ldots,s} |x_j|$. The dynamic range of $x$ is defined as $\xmax/\xmin$. 
}

\section{Stability analysis}
\label{seccon}
\commentout{
For  a given matrix $A$, let $\smax(A)$ and $\smin(A)$ denote  the maximum and minimum {\em nonzero} singular values of $A$, respectively. Denote the spectral norm and Frobenius norm of $A$ by $\|A\|_2$ and $\|A\|_F$, respectively. 
}

First we have 
\beq
\nn
\|\hat\Psi- \Psi\|_2&\leq &\|( \hat H_1^\dagger-H_1^\dagger)H_2^\ep\|_2+
\|H_1^\dagger (H_2^\ep-H_2)\|_2\\
&\le& \|\hat H_1^\dagger-H_1^\dagger\|_2\|H_2^\ep\|_2+
\|H_1^\dagger\|_2\| E_2\|_2.\label{29}
\eeq

Suppose at first $\|E_1\|_2< \sigma_s$ (to be justified later) so that by Weyl's theorem $\sigma_s^\ep>0$ and $\rank(\hat H_1)=\rank(H_1)$. 
Wedin's inequality (\cite{SS}, Theorem III.3.8) asserts  that 
 \commentout{
\beq
\label{wed}
\|\hat H_1^\dagger-H_1^\dagger\|_{\rm F} \leq \max\{\|\hat H_1^\dagger\|_2^2, \|H_1^\dagger\|_2^2\} \|\hat H_1-H_1\|_{\rm F}
\eeq
}
\beq
\label{wed}
\|\hat H_1^\dagger-H_1^\dagger\|_2 \leq {1+\sqrt{5}\over 2}\|\hat H_1^\dagger\|_2 \|H_1^\dagger\|_2\|\hat H_1-H_1\|_2
\eeq
\commentout{
where $\mu$ is given in the following table
\beq
\label{table}
\begin{tabular}{c|c|c|c|}
$\|\cdot\|$ & \hbox{\rm arbitrary}&\hbox{\rm spectral}& \hbox{\rm Frobenius}\\
\hline
$\mu$&3&${1+\sqrt{5}\over 2}$&$ \sqrt{2}$.
\end{tabular}
\eeq
For our
purpose, we are primarily concerned with the spectral norm $\|\cdot\|_2$ and the Frobenius norm $\|\cdot\|_{\rm F}$.
}
First let us estimate $\|\hat H_1-H_1\|_2$. We have
\beqn
\|\hat H_1-H_1\|_2&=&\|\cP_s H_1^\ep-H_1\|_2\\
&\leq &\|\cP_s  H_1^\ep-H_1^\ep\|_2+\|H_1^\ep-H_1\|_2\\
&= &\|(I-\cP_s)  H_1^\ep\|_2+\|E_1\|_2
\eeqn
where $\cP^\perp_s= I-\cP_s$ is the projection onto the ``noise" subspace of $H_1^\ep$. 
Hence
\beqn
\|\cP^\perp_s H_1^\ep\|_2&=&\sigma^\ep_{s+1}
=\sigma^\ep_{s+1}-\sigma_{s+1}
\leq\|E_1\|_2
\eeqn
by Weyl's theorem (\ref{Weyl}). 
Therefore Wedin's bound (\ref{wed}) becomes
\beq
\label{wed2}
\|\hat H_1^\dagger-H_1^\dagger\|_2\leq 2\|\hat H_1^\dagger\|_2\|H_1^\dagger\|_2 \|E_1\|_2
\eeq
and consequently the bound (\ref{29}) becomes
\beq
\|\hat\Psi- \Psi\|_2&\leq &\|H_1^\dagger\|_2\big(2\|\hat H_1^\dagger\|_2 \|H_2^\ep\|_2 \|E_1\|_2
+\| E_2\|_2\big)\equiv \eta.\label{31}
\eeq

Next we use the discrete Ingham inequalities 
to estimate 
\[
\|H_1^\dagger\|_2=\sigma_s^{-1}, \quad \|\hat H_1^\dagger\|_2=(\sigma^\ep_s)^{-1},\quad \|H^\ep_2\|_2=\sigma_1^\ep. 
\]
The discrete Ingham inequalities are extension of the continuum version first proved in 
\cite{Ingham} (see also \cite{Young}).

\begin{theorem} \cite{my-music} Let $N$ be any integer. 
If $\supp$ satisfies  the separation condition 
\beqn
\label{sep}
\delta=\min_{j\neq l} d(\om_j,\om_l) > \frac 1 N\Big(1 - {2\pi \over N} \Big)^{-\frac 1 2}
\eeqn
then for  any $ z\in \CC^s$
\beq
{\|\Phi^N z\|_2^2\over \|z\|_2^2} \ge  N\Big(\frac{2}{\pi} - \frac{2}{\pi N^2 \delta^2}-\frac 4 N\Big). \label{35}
\eeq
Moreover, when $N$ is even 
\beq
 {\|\Phi^N {z}\|_2^2 \over \|z\|_2^2}
 \le 
 N\Big(\frac{4\sqrt 2}{\pi}  + \frac{\sqrt 2}{\pi N^2 \delta^2} + \frac{3\sqrt 2}{N}\Big)
 \label{36}
\eeq
 and when  $N$ is odd
\beq
\label{37}
{\|\Phi^N {z}\|_2^2\over \|z\|_2^2}
\le
\left(N+1\right)\left(\frac{4\sqrt 2}{\pi}  + \frac{\sqrt 2}{\pi (N+1)^2 \delta^2}+ \frac{3\sqrt 2}{N+1} \right). 
\eeq
\label{thm4}
\end{theorem}

By the Vandermonde decomposition (\ref{H1}) for $H_1$,  Theorem \ref{thm4} with
$N=L-1, M-L$, immediately implies the following. 
\begin{corollary} Under the separation condition
\beq
\delta> \max\Big\{{1\over L-1} \Big(1 - {2\pi \over L-1} \Big)^{-\frac 1 2}, {1\over M-L} \Big(1 - {2\pi \over M-L} \Big)^{-\frac 1 2} \Big\}\label{sep3}
\eeq
we have 
\beq
\sigma_s &\ge &{2\xmin \over \pi}{
\left( {L-1} - \frac{1}{(L-1) \delta^2} - {2\pi}\right)^{1/2}
\left( {M-L} - \frac{1}{(M-L) \delta^2} - {2\pi}\right)^{1/2}}\label{22}\\
\sigma_1&\le& {4\sqrt{2}\xmax\over \pi} \left({L}  + \frac{1}{4 L \delta^2}+ {3\pi\over 4} \right)^{1/2}\left(M-L+1 + \frac{1}{4(M-L+1) \delta^2}+ {3\pi\over 4} \right)^{1/2}\label{24}
\eeq
where 
\[
\xmin=\min_j\{|x_j|\}, \quad \xmax=\max_j\{|x_j|\}. 
\]

\end{corollary}
\begin{remark}
For a fixed $M$, the right hand side of (\ref{22}) can be maximized by $L=\big[{M+1\over 2}\big]$, the largest integer not greater than ${M+1\over 2}$, under which the separation condition (\ref{sep3}) becomes
\beq
\delta>  \ell\equiv {2 \over M}\Big(1 - {4\pi \over M} \Big)^{-\frac 1 2}\label{sep2}
\eeq
and hence
the bounds (\ref{22})-(\ref{24}) become
\beq
\sigma_1&\le& {\xmax M 2\sqrt 2\over \pi}\left(1+ \frac{1}{M^2 \delta^2}+ {2\over M}+{3\pi\over 2M} \right)\label{30}\\
\sigma_s &\ge &{\xmin M\over \pi }\left(1 - \frac{4}{ M^2\delta^2} - {4\pi \over M}\right).
\label{23}
\eeq
\end{remark}

Finally to  justify the assumption  $\sigma_s>\|E_1\|_2$ it suffices to have
\beq
\|E_1\|_2<{\xmin M\over \pi }\left(1 - \frac{4}{ M^2\delta^2} - {4\pi \over M}\right)
\eeq
under  (\ref{sep2}) (which renders the right hand side positive).

By (\ref{30})-(\ref{23}) and Weyl's theorem, we obtain 
\beq
\|H_1^\dagger\|_2&\leq& {\pi\over \xmin M }\left(1 - \frac{4}{ M^2\delta^2} - {4\pi \over M}\right)^{-1}\label{46}\\
\|\hat H_1^\dagger\|_2&\leq& \left({\xmin M \over \pi }\left(1 - \frac{4}{ M^2\delta^2} - {4\pi \over M}\right)-\|E_1\|_2\right)^{-1}\label{47}\\
\|H^\ep_2\|_2&\leq & {\xmax M 2\sqrt 2\over \pi}\left(1+ \frac{1}{M^2 \delta^2}+ {2\over M}+{3\pi\over 2M} \right)\label{48}
\eeq
which lead to the corresponding bound on $\eta$ via (\ref{31}). 

Summarizing the preceding analysis, we have the following theorem
\begin{theorem}
Let $\rho=\delta M$ be the minimum separation in the unit of RL.
Under the separation condition (\ref{sep2}),  or  equivalently
\beq
\label{sep4}
 \rho> 2\Big(1 - {4\pi \over M} \Big)^{-\frac 1 2}, 
\eeq
and
\[
\|E_1\|_2<
{\xmin M \over \pi }\left(1 - \frac{4}{ M^2\delta^2} - {4\pi \over M}\right)
\]
we  have 
\beqn
\|\hat\Psi-\Psi\|_2 &\leq &\|H_1^\dagger\|_2\big(2\|\hat H_1^\dagger\|_2 \|H_2^\ep\|_2 \|E_1\|_2
+\| E_2\|_2\big)\equiv \eta\label{51}
\eeqn
with an upper bound given by  (\ref{46})-(\ref{48}). In particular, for  $ M\gg 1$, $\eta$ has the asymptotic
\beqn
\eta\approx {\pi\over \xmin(1-4\rho^{-2})}\left[ {4\sqrt{2}(1+\rho^{-2})\xmax\over(1-4\rho^{-2}) \xmin-\pi\|E_1\|_2/M}{\|E_1\|_2\over M}+
{\|E_2\|_2\over M}\right].
\eeqn

\label{thm}
\end{theorem}

As noted before,  the spectral norm of a random Hankel
 matrix from a zero mean, i.i.d.  sequence of a finite variance is on the order of $\sqrt{M \log M}$ \cite{Adam}. Therefore for i.i.d. noise the error bound in Theorem \ref{thm}  tends
 to zero like $\sqrt{\log M/M}$ with a constant depending on the dynamic range $\xmax/\xmin$ and the minimum separation $\rho>2$ in the unit of RL. 
 
Now we are ready to use Elsner's theorem (\cite{SS}, Theorem IV.1.3) to conclude
\beq
\mu_{\rm H} (\hat\Psi, \Psi)
\leq \Big(\|\hat\Psi\|_2+\|\Psi\|_2\Big)^{1-{1\over M-L+1}}\|\hat\Psi-\Psi\|^{1\over M-L+1}_2\label{53}
\eeq
where \[
\mu_{\rm H}(\hat\Psi, \Psi)=\max\big\{\max_{i} \min_j |\hat \lambda_i-\lambda_j|,\,\,
\max_{j} \min_i |\hat \lambda_i-\lambda_j|\big\}
\]
is  the Hausdorff Metric (HM) of the two sets of eigenvalues in question. Bound (\ref{53}) can be
made more concrete by using Theorem \ref{thm} and the fact $\|\Psi\|_2=1$:
\beq
\mu_{\rm H}(\hat\Psi, \Psi)
\leq \Big(2+\eta  \Big)^{1-{1\over M-L+1}}\eta^{1\over M-L+1}\label{54}
\eeq

\commentout{
Now we are ready to use the Bauer-Fike theorem \cite{BF} to bound the differences in
eigenvalues in terms of the difference in the spectral norm. To this end, we need to construct
an invertible matrix $\widetilde Y$ which diagonalizes $\Psi$ i.e.
\beq
\label{33}
\Psi=\widetilde Y^{-1} \widetilde \Lambda \widetilde Y
\eeq
where $\widetilde \Lambda=\text{diag}(e^{i-2 \pi \omega_1}, \dots, e^{-i2 \pi \omega_s}, 0,\dots, 0)\in \CC^{(M-L+1)\times (M-L+1)}$. Then the Bauer-Fike theorem asserts
\beq
\label{44}
\min_{\lambda}|\lambda-\hat \lambda|\leq \|\widetilde Y\|_2\|\widetilde Y^{-1}\|_2\|\hat\Psi-\Psi\|_2,\quad \lambda\in \{0, e^{-i2\pi \omega_1},\dots, e^{-i2\pi \omega_s}\}
\eeq
where $\hat \lambda$ is any eigenvalue of $\hat \Psi$ and $F$ is given in (\ref{31}).
{\color{red} This bound is useless as we can't rule out the possibility that $\hat \lambda$ may be clustered.
It would have been much better if we can approximate $e^{-i2\pi \om_j}$ by
$\hat\lambda$ which requires the diagonalizability of $\hat \Psi$.}

Let the SVD of $Y$ be written as
\[
Y={A} \ \underbrace{\text{diag}(\eta_1,\eta_2,\ldots,\eta_s)}_{s\times (M-L+1)} B^\star
\]
and let $C$ be any $(M-L+1-s)\times (M-L+1-s)$ unitary matrix. 
Let
\[
\widetilde Y=\left[\begin{matrix}{A} &0\\
0&C
\end{matrix}\right] \ \underbrace{\text{diag}(\eta_1,\eta_2,\ldots,\eta_s,\eta_s,\eta_s,\dots,\eta_s)}_{(M-L+1)\times (M-L+1)} \ B^\star.
\]
It is straightforward to verify eq. (\ref{33}) using the identify $\Psi=Y^\dagger \Lambda Y$.

We also have $\|\widetilde Y\|_2=\eta_1=\|Y\|_2, \|\widetilde Y^{-1}\|_2=\eta_s^{-1}=\|Y^\dagger\|_2$. On the other hand, in view of the definition (\ref{11}), the discrete Ingham inequalities (\ref{35})-(\ref{37}) with $N=M-L+1$  are applicable here and give
\begin{corollary} Under the separation condition (\ref{sep}) 
\beq
\eta_1&\le&\xmax \sqrt{4\sqrt{2}\over \pi}\left(M-L+1 + \frac{1}{4(M-L+1) \delta^2}+ {3\pi\over 4} \right)^{1/2}\label{40}\\
\eta_s &\ge &\xmin\sqrt{2\over \pi}
\left( {M-L} - \frac{1}{(M-L) \delta^2} - {2\pi}\right)^{1/2}\label{41}\\
\eeq
with $\xmin=\min_j\{x_j\}, \xmax=\max_j\{x_j\}$. 
\label{cor2}
\end{corollary}
Corollary \ref{cor2} leads immediately to the bound on the condition number $\kappa_2$  of $\widetilde Y$
\beq
\kappa_2=\|Y\|_2\|Y^\dagger\|_2\leq \sqrt{2\sqrt{2}}{\xmax\over \xmin}\left(M-L+1 + \frac{1}{4(M-L+1) \delta^2}+ {3\pi\over 4}\over {M-L} - \frac{1}{(M-L) \delta^2} - {2\pi}\right)^{1/2}\label{cond}.
\eeq

Combining  (\ref{31}) and (\ref{44}), we have 
\beq
\min_{\lambda}|\lambda-\hat \lambda|\leq \kappa_2\|F\|_2,\quad \lambda\in \{0, e^{-i2\pi \omega_1},\dots, e^{-i2\pi \omega_s}\}\label{49}
\eeq
where $\kappa_2$ is bounded via (\ref{cond}) and $Y$ is bounded via (\ref{46})-(\ref{48}). 
\section{Conclusion}
The full expression of the bound (\ref{49}) may be complicated but its meaning is
not hard to grasp if we let $M\to \infty$ with $L=\Big[{M+1\over 2}\Big]$. In this case,
if the true frequencies are separated by slightly more than $2/M$ (2 RL),
say $\delta M=\rho>2$, 
then every eigenvalue of SS-ESPRIT estimate $\hat \Psi$ lies
with the radius 
\beq
\label{50}
\sqrt{2\sqrt{2}}{\xmax\over \xmin}\sqrt{\rho^2+1\over \rho^2-4} {\pi\over \xmin(1-4\rho^{-2})}\left[ {4\sqrt{2}(1+\rho^{-2})\xmax\over(1-4\rho^{-2}) \xmin-\pi\|E_1\|_2/M}{\|E_1\|_2\over M}+
{\|E_2\|_2\over M}\right]
\eeq
of some true Fourier mode $e^{-i2\pi \omega_j}$. As noted before,  the spectral norm of a random Hankel
 matrix from a zero mean, i.i.d.  sequence of a finite variance is on the order of $\sqrt{M \log M}$ \cite{Adam}. Therefore for i.i.d. noise the error bound (\ref{50}) tends
 to zero like $\sqrt{\log M/M}$ with a constant depending on the dynamic range $\xmax/\xmin$ and the minimum separation $\rho>2$ in the unit of RL. 
}

\section{Conclusion}
\begin{figure}[t]
\begin{minipage}[b]{0.5\linewidth}
\centering
\includegraphics[width=\textwidth]{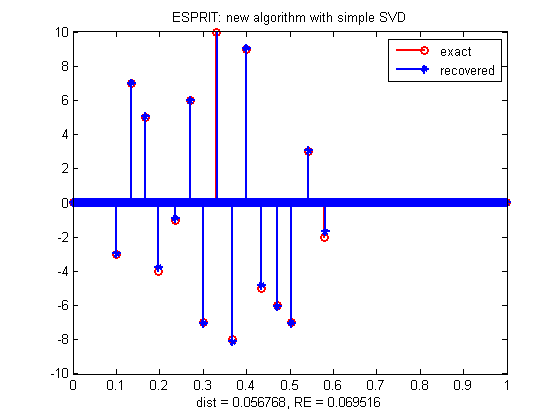}
\label{fig:figure1}
(a) ESPRIT: $\mu_{\rm H}(\hat{S},S) = 0.057$RL. 
\end{minipage}
\hspace{0.5cm}
\begin{minipage}[b]{0.5\linewidth}
\centering
\includegraphics[width=\textwidth]{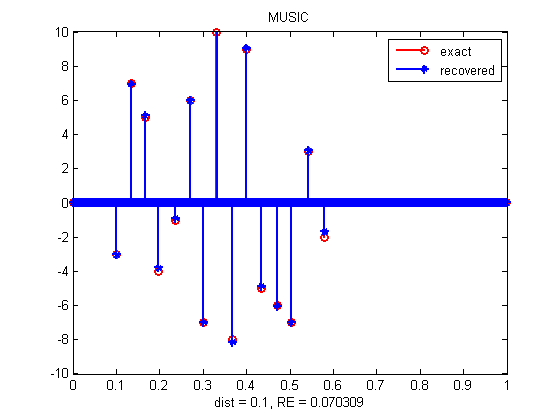}
\label{fig:figure2}
(b) MUSIC: $\mu_{\rm H}(\hat{S},S) = 0.1$RL. 
\end{minipage}
\commentout{
\begin{minipage}[b]{0.5\linewidth}
\centering
\includegraphics[width=\textwidth]{BLOOMP_34_noisy.png}
\label{fig:figure1}
(a)  BLOOMP: $\mu_{\rm H}(\hat{S},S) = 0.133$RL.  
\end{minipage}
}
\caption{Reconstruction of 15 real-valued frequencies separated by 3-4 RL with $10\%$ NSR.} \label{fig2}
\end{figure}

\begin{figure}[h]
\centering
\includegraphics[width=15cm]{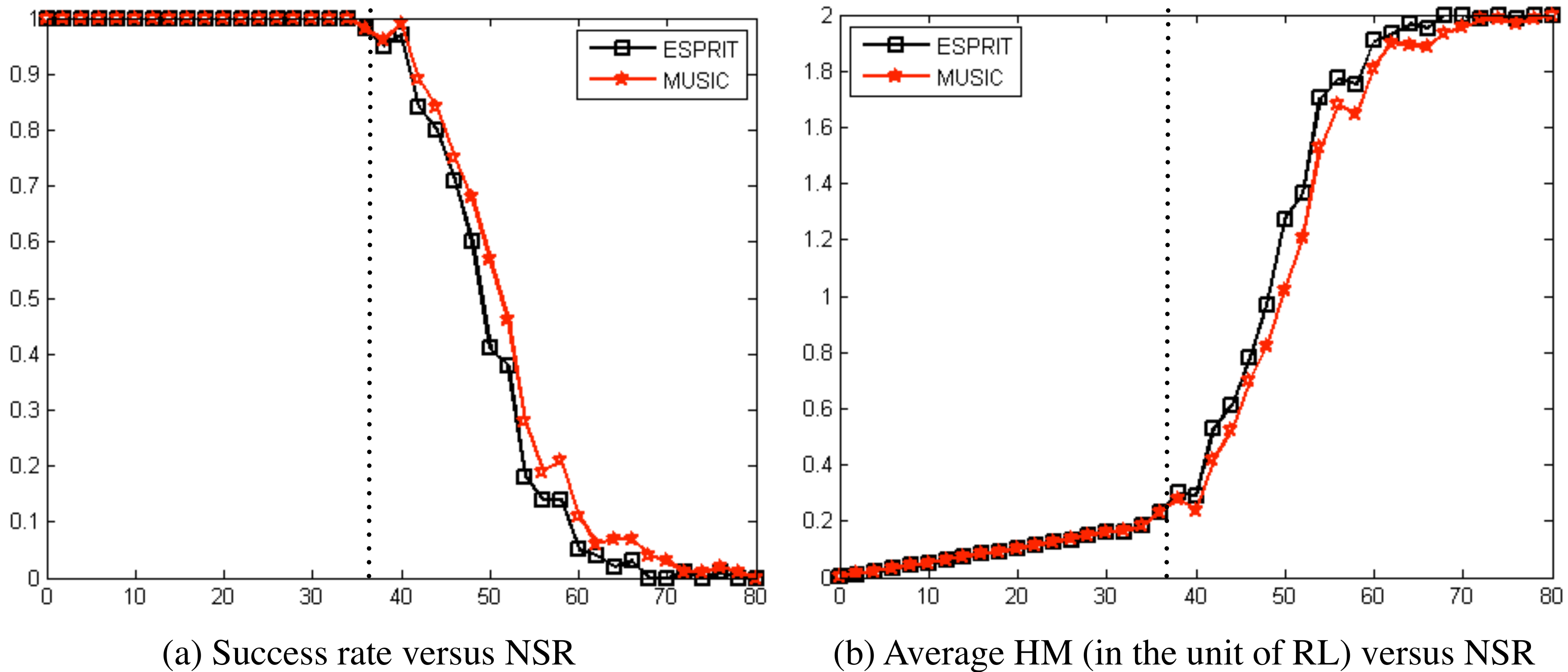}
\commentout{
\subfigure[HM for  separation $2 \sim  3$  RL]{
\includegraphics[width=7cm]{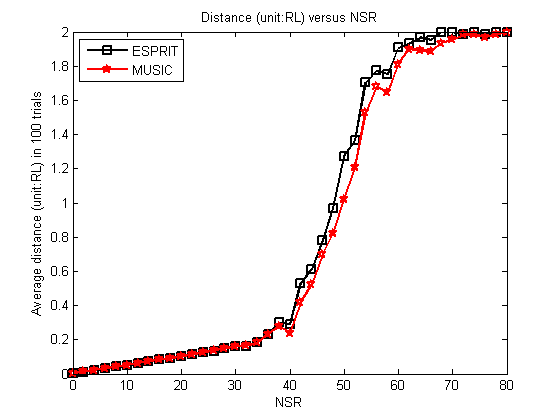}
}
\subfigure[Success rate for separation $2 \sim 3$ RL]{
\includegraphics[width=7cm]{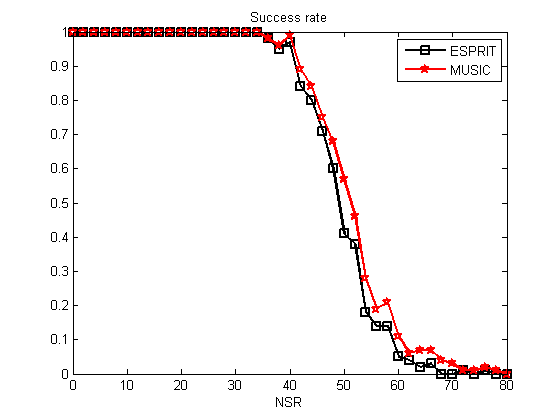}
}
\subfigure[Success rate for  separation $4\sim 5$ RL]{
\includegraphics[width=7cm]{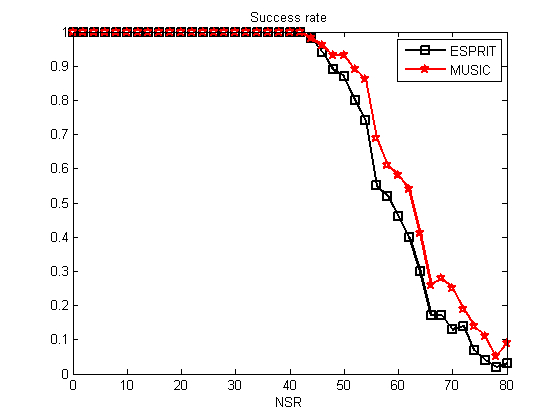}
}
\subfigure[HM for separation $4 \sim  5$ RL]{
\includegraphics[width=7cm]{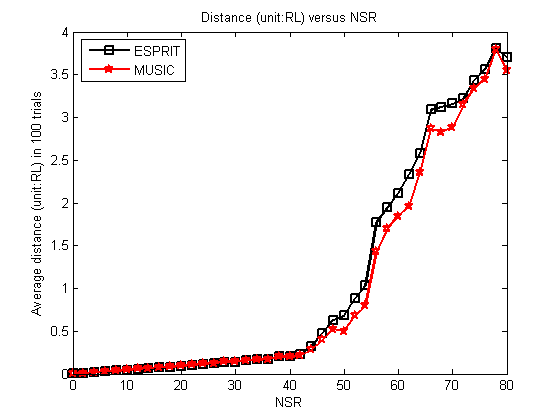}
}
}
\caption{(a) Success rate and (b) average HM vs. NSR (in percentage) 
for separation of 2-3 RL}
\label{fig1}
\end{figure}

In conclusion, we have given performance guarantees for
SS-ESPRIT. 
In particular, for noiseless measurement with $M+1\geq 2s$, Theorems \ref{thm0} and \ref{thm1} guarantee exact recovery for any subset $\cS\subset [0,1]$ of $s$ frequencies. 
For noisy measurement, Theorem \ref{thm} guarantees
noise stability under the separation condition
\beqn
\rho>  2\Big(1 - {4\pi \over M} \Big)^{-\frac 1 2} 
\eeqn
in the unit of RL. 
This separation and sparsity constraint compares favorably with 
those of other approaches to compressed sensing in the continuum
which are at least 3-4  RL
\cite{Csr2,Csr1,Chi, DNN,FL, Tang}. 

Numerical simulation demonstrates stability to a significant level of noise. The noise $\ep$ is additive i.i.d. complex Gaussian, i.e. $\epsilon \sim N(0, \nu^2 I) + i N(0, \nu^2 I)$ of various strength in terms of 
 the Noise-to-Signal Ratio (NSR)  
$$\text{NSR} = \mathbb{E}\{||\ep||_2\}/||y||_2 = \nu \sqrt{2(M+1)}/||y||_2$$
where $M=100$. 
The error metric the Hausdorff metric $\mu_{\rm H}(\cS, \widehat \cS)$  between the exact $\mathcal{S}$ and recovered $\widehat{\mathcal{S}}$ sets of frequencies. 
\commentout{
\begin{equation} 
\mu_{\rm H}(\hat{\mathcal{S}}, \mathcal{S}) = \max \left \{ \max_{\hat{\omega} \in \hat{\mathcal{S}}} \min_{\omega \in \mathcal{S}} d(\hat{\omega},\omega), \max_{\omega \in \mathcal{S}} \min_{\hat{\omega} \in \hat{\mathcal{S}}} d(\hat{\omega},\omega) \right \}.
\end{equation}
}
We use two reconstruction methods: SS-ESPRIT analyzed above and 
MUSIC studied in \cite{my-music} (see also \cite{Cheney,CheneyRadar,MUSIC,Kirsch})  both of which
employ the Hankel matrix (\ref{hankel}) and the Vandermonde decomposition (\ref{van}).

Fig. \ref{fig2} shows an instance of reconstruction of 15 frequencies that are randomly distributed,  separated
by 3-4 RL and have real-valued amplitudes of   dynamical range $\xmax/\xmin=10$, from
$M=100$ measured data of $10\%$ NSR. Both ESPRIT and MUSIC perform well with comparable accuracy.

For Fig. \ref{fig1}, the frequency set $\cS$ consists of 20 randomly selected frequencies separated by $2-3$ RL, with randomly phased amplitudes $x$  of equal strength (i.e. the dynamic range $\xmax/\xmin=1$). A reconstruction is {\em successful}  if  $\mu_{\rm H}(\widehat{\mathcal{S}},\mathcal{S}) \le 1$RL.

Fig.\ref{fig1}(a)  shows the success rate for 100 independent trials versus NSR.  Clearly a ``phase transition" occurs at the threshold NSR $\approx 37\%$ beyond which  the success rate begins to drop precipitously. The threshold NSR depends on the frequency spacings, the numbers of  data and frequencies as well as the dynamic range. 

Fig.\ref{fig1}(b) shows $\mu_{\rm H} (\cS,\widehat\cS)$,  averaged over 100 independent trials, versus NSR  and  exhibits the same
phase transition where the rapid growth of $\mu_{\rm H}$ is due to reconstruction failure. Notably  the average $\mu_{\rm H}$ below the threshold  does not exceed $0.2$RL, much better than the success criterion of 
$1$RL. 

Again the performances of ESPRIT and MUSIC 
 are comparable in Fig. \ref{fig1}  with 
the main difference being the speed of computation: SS-ESPRIT
is about ten times faster than MUSIC in our simulation.

\bigskip
{\bf Acknowledgements.} Research is supported in part by  US NSF grant DMS-1413373 and Simons Foundation grant 275037. I thank Lu Li for help in preparing the figures.
\bigskip

\bibliographystyle{plain}	
\bibliography{myref}		
 
\commentout{
\appendixtitleon
\appendixtitletocon
\begin{appendices}
\section{Proof of Theorem \ref{thm3}}
\label{app2}
Fixing $\supp = \{\om_1,\ldots,\om_s\}\subset \TT$, we define the matrix $\Phi^{N_1 \rightarrow N_2}$ such that $$\Phi^{N_1 \rightarrow N_2}_{kj} = e^{-2\pi i k \om_j}, \ k = N_1,\ldots,N_2,\ j = 1,\ldots,s.$$
For simplicity, we denote $\Phi^M = \Phi^{0 \rightarrow M}.$

Let 
\beq
G(\om) = \sum_{k = 0}^{L} g(\frac k L)e^{2\pi i k \om},\quad g(t) = \cos \pi (t-0.5).
\label{eqG}
\eeq
Function $G$ has the following properties.
\begin{lemma}
\label{lemma4}
\begin{enumerate}
   \setlength{\itemsep}{0.2cm}
\item $G(\om+n) = G(\om )$ for $n \in \ZZ$.

\item $G(-\om) = e^{-2\pi i L\om}G(\om)$ and $|G(-\om)| = |G(\om)|$.

\item $L(\frac 2 \pi - \frac{1}{L}) \le G(0) \le L(\frac 2 \pi+ \frac{1}{L}).$

\item $
|G(\om)| \le\frac{2}{\pi}\frac{L}{|1-4L^2\om^2|} + \frac{8}{\pi L}$ for $\om \in [0,1/2]$.
\end{enumerate}
\end{lemma}

\begin{proof}
\begin{enumerate}
\item 
$$G(\om+n) =  \sum_{k=0}^L g(\frac k L) e^{2\pi i k(\om+n)}= \sum_{k=0}^L g(\frac k L) e^{2\pi i k\om} = G(\om).$$

\item 
\begin{align*}
G(-\om) & = \sum_{k=0}^L g(\frac k L) e^{-2\pi i k \om} 
 = \sum_{l=0}^L g\left(\frac{L-l}{L}\right) e^{-2\pi i(L-l)\om} \text{ by letting } l = L-k \\
 & = \sum_{l=0}^L \left[\cos\pi (1-\frac l L - 0.5)\right] e^{-2\pi i(L-l)\om} 
 = \sum_{l=0}^L \left[\cos\pi (\frac l L - 0.5)\right] e^{-2\pi i(L-l)\om}
 \\
 & = \sum_{l=0}^L g(\frac l L) e^{-2\pi i(L-l)\om} = e^{-2\pi i L\om} \sum_{l=0}^L g(\frac l L) e^{2\pi i l \om} 
 = e^{-2\pi i L\om} G(\om).
\end{align*}

\item 
On the one hand, $$\|g\|_1 = \int_{0}^{1}\cos \pi (t-0.5)dt = \frac 2 \pi.$$ 
On the other hand,
$$G(0) = \displaystyle \sum_{0}^{L} g(\frac k L),$$ and $$\frac{1}{L}(G(0) - 1)\le \|g\|_1 \le  \frac{1}{L}(G(0) + 1).$$

\item 
According to the Poisson summation formula,
\begin{align}
& \frac 1 L G(\om) = \frac{1}{L} \sum_{k= 0}^{L} g(\frac{k}{L})e^{2\pi i k \om}
 = \sum_{r = -\infty}^{\infty} \int_{0}^{1} g(z)e^{2\pi i L(\om-r)z}dz
 = \frac{2}{\pi}\sum_{r=-\infty}^{\infty} \frac{\cos{\pi L(\om-r)}}{1-4L^2(\om-r)^2} e^{i\pi L(\om-r)}
\label{eqpoisson}.
\end{align}
Hence
\begin{align*}
& \Big|\frac 1 L G(\om)\Big| \le \frac{2}{\pi}\sum_{r=-\infty}^{\infty}\Big| \frac{\cos{\pi L(\om-r)}}{1-4L^2(\om-r)^2}\Big| \\
& \le \frac{2}{\pi}\frac{1}{|1-4L^2\om^2|} + \frac{2}{\pi} \sum_{r \neq 0} \frac{1}{4L^2(r-\om)^2-1} \\
& \le \frac{2}{\pi}\frac{1}{|1-4L^2\om^2|} + \frac{2}{\pi} \sum_{r \neq 0} \frac{2}{4L^2(r-\om)^2} \\
 & \le \frac{2}{\pi}\frac{1}{|1-4L^2\om^2|} + \frac{4}{\pi} \Big[\frac{1}{4L^2(\frac 1 2)^2}+\frac{1}{4L^2(1)^2} + \frac{1}{4L^2(\frac 3 2)^2} + \frac{1}{4L^2(2)^2} +\ldots\Big] \\
 & \hspace{5cm}\text{ as } \om \in [0,\frac 1 2],
  \\
  & \le  \frac{2}{\pi}\frac{1}{|1-4L^2\om^2|} + \frac{4}{\pi}  \frac{1}{L^2}\Big[\frac{1}{1^2}+\frac{1}{2^2} + \frac{1}{3^2} + \frac{1}{4^2} +\ldots\Big]  \\
  & \le  \frac{2}{\pi}\frac{1}{|1-4L^2\om^2|} + \frac{4}{\pi} \frac{1}{L^2} 2 =  \frac{2}{\pi}\frac{1}{|1-4L^2\om^2|} + \frac{8}{\pi L^2}.\\
\end{align*}

In \eqref{eqpoisson} the difference between the discrete and the continuous case lies in $$\frac{2}{\pi}\sum_{r \neq 0} \frac{\cos{\pi L(\om-r)}}{1-4L^2(\om-r)^2}$$
which is bounded above by $8/(\pi L^2)$, and is therefore negligible when $L$ is sufficiently large.
 
\end{enumerate}
\end{proof}
The following lemma paves the way for the proof of Theorem \ref{thm3}.
\begin{lemma}
\label{lemma6}
Suppose objects in $\supp$ satisfy the separation condition 
\beq
\label{lemma60}
d(\om_j,\om_l) \ge q > \frac 1 L \sqrt{\frac 2 \pi}\left( \frac 2 \pi - \frac 1 L - \frac{8s}{\pi L^2}\right)^{-\frac 1 2}.
\eeq
Then
\beq
\Big(\frac{2}{\pi}-\frac 1 L - \frac{2}{\pi L^2 q^2} - \frac{8s}{\pi L^2}\Big)\|{{\bc}}\|_2^2
\le
\frac{1}{L}\sum_{k= 0}^{L} g(\frac{k}{L}) \Big|(\Phi^{0\rightarrow L}  {\bc})_k\Big|^2 
\le 
\Big(\frac{2}{\pi}+\frac 1 L + \frac{2}{\pi L^2 q^2} + \frac{8s}{\pi L^2}\Big)\|{\bc}\|_2^2\label{lemma61}
\eeq
for all ${\bc} \in \CC^s$.
\end{lemma}

\begin{proof}
\begin{align*}
&\sum_{k=0}^{L} g(\frac{k}{L}) \Big|(\Phi^{0\rightarrow L}  {\bc})_k\Big|^2  
= \sum_{k= 0}^{L} g(\frac{k}{L}) 
\overline{\sum_{j=1}^s {\bc}_j e^{-2\pi i k \om_j}}
\sum_{l=1}^s {\bc}_l e^{-2\pi i k \om_l}\\
& = \sum_{j=1}^s \sum_{l=1}^s \overline{{\bc}_j} {\bc}_l \sum_{k=0}^{L} g(\frac k L)e^{2\pi i k(\om_j -\om_l)} = \sum_{j=1}^s\sum_{l=1}^s G(\om_j-\om_l) \overline{{\bc}_j} {\bc}_l \\
& = G(0) \|{\bc}\|_2^2 + \sum_{j=1}^s \sum_{l\neq j} G(\om_j-\om_l)\overline{{\bc}_j}{\bc}_l.
\end{align*}
It follows from the triangle inequality that
\begin{align*}
 G(0) \|{\bc}\|_2^2 - \sum_{j=1}^s \sum_{l\neq j} |G(\om_j-\om_l)\overline{{\bc}_j}{\bc}_l|
\le
\sum_{k= 0}^{L} g(\frac{k}{L}) \Big|(\Phi^{0\rightarrow L }  {\bc})_k\Big|^2 \le G(0) \|{\bc}\|_2^2 + \sum_{j=1}^s \sum_{l\neq j} |G(\om_j-\om_l)\overline{{\bc}_j}{\bc}_l|, \\
\end{align*}
where $\displaystyle \sum_{j=1}^s \sum_{l\neq j} |G(\om_j-\om_l)\overline{{\bc}_j}{\bc}_l|$ can be estimated through Property 4 in Lemma \ref{lemma4}.
\begin{align}
&\sum_{j=1}^s \sum_{l\neq j} |G(\om_j-\om_l)\overline{{\bc}_j}{\bc}_l| 
 \le \sum_{j=1}^s \sum_{l\neq j} |G(\om_j-\om_l)|\frac{|{\bc}_j|^2+|{\bc}_l|^2}{2}  = \sum_{j=1}^s |{\bc}_j|^2 \sum_{l \neq j} |G(\om_j-\om_l)|
 \nonumber
\\
& = \sum_{j=1}^s |{\bc}_j|^2 \sum_{l \neq j} |G(d(\om_j,\om_l))|
\le \sum_{j=1}^s |{\bc}_j|^2 \sum_{l \neq j} \Big[\frac{2}{\pi} \frac{L}{|1-4L^2d^2(\om_j,\om_l)|} +\frac{8}{\pi L}\Big] 
\nonumber 
\\
& \le \sum_{j=1}^s |{\bc}_j|^2 \frac{4}{\pi} \sum_{n = 1}^{\lfloor \frac{s}{2}\rfloor} \Big[ \frac{L}{4 L^2 n^2 q^2 -1} +\frac{4}{ L} \Big]  \text{ as frequencies in } \supp \text{ are pairwise separated by } q >\frac{1}{L} 
\nonumber \\
& \le \sum_{j=1}^s |{\bc}_j|^2  \frac{4}{\pi} \Big[\frac{2 s}{ L} + \sum_{n=1}^\infty \frac{L}{4L^2 n^2 q^2 -1}\Big] 
\le \sum_{j=1}^s |{\bc}_j|^2  \frac{4}{\pi} \Big[\frac {2s}{ L} +  \frac{1}{L^2 q^2}\sum_{n=1}^\infty\frac{L}{4 n^2  -1}\Big]
\label{lemma62} \\
& \le \sum_{j=1}^s |{\bc}_j|^2  \frac{4}{\pi} \Big[\frac {2s}{ L} +  \frac{1}{L q^2} \frac{1}{2}\sum_{n=1}^\infty\Big(\frac{1}{2n-1}-\frac{1}{2n+1}\Big)\Big] 
=\sum_{j=1}^s |{\bc}_j|^2  \frac{4}{\pi} \Big[\frac {2s}{ L} +  \frac{1}{L q^2} \frac{1}{2}\Big] 
\nonumber\\
& = \|{\bc}\|_2^2  \frac{2}{\pi} \Big(\frac{1}{L q^2} + \frac{4s}{L}\Big),
\nonumber
\end{align}
where  $\lfloor {s\over 2} \rfloor$  denotes the nearest integer smaller than or equal to $s/2$. Kernel $g$ in \eqref{lemma61} is crucial for the convergence of the series in \eqref{lemma62}.

Therefore 
$$G(0) \|{\bc}\|_2^2 - \frac{2}{\pi} \Big(\frac{1}{L q^2} + \frac{4s}{L}\Big)\|{\bc}\|_2^2 
\le
\sum_{k= 0}^{L} g(\frac{k}{L}) \Big|(\Phi^{0\rightarrow L}  {\bc})_k\Big|^2 \le G(0) \|{\bc}\|_2^2 + \frac{2}{\pi} \Big(\frac{1}{L q^2} + \frac{4s}{L}\Big)\|{\bc}\|_2^2 . $$

The equation above along with Property 3 in Lemma \ref{lemma4} yields
$$
L\Big(\frac{2}{\pi}-\frac 1 L - \frac{2}{\pi L^2 q^2} - \frac{8s}{\pi L^2}\Big)\|{\bc}\|_2^2
\le
\sum_{k= 0}^{L} g(\frac{k}{L}) \Big|(\Phi^{0\rightarrow  L}  {\bc})_k\Big|^2 
\le 
L\Big(\frac{2}{\pi}+\frac 1 L + \frac{2}{\pi L^2 q^2} + \frac{8s}{\pi L^2}\Big)\|{\bc}\|_2^2. $$

The separation condition \eqref{lemma60} is derived from the positivity condition of the lower bound, i.e.,
$$
\frac{2}{\pi}-\frac 1 L - \frac{2}{\pi L^2 q^2} - \frac{8s}{\pi L^2} > 0.$$
\end{proof}

Proof of Theorem \ref{thm3} is given below.

\begin{proof}
Given that $\supp = \{\om_1,\ldots,\om_s\} \subset [0,1)$ and frequencies are separated above $1/L$, there are no more than $L$ frequencies in $\supp$, i.e., $s<L$.

The lower bound in Theorem \ref{thm3} follows from Lemma \ref{lemma6} as  
\begin{align*} 
&\|\Phi^{0\rightarrow L}{\bc}\|_2^2 
=
\sum_{k = 0}^{L} \Big|(\Phi^{0\rightarrow L}  {\bc})_k\Big|^2 
=
\sum_{k= 0}^{L}  \Big|(\Phi^{0\rightarrow L}  {\bc})_k\Big|^2
\ge \sum_{k= 0}^{L} g(\frac{k}{L}) \Big|(\Phi^{0\rightarrow L}  {\bc})_k\Big|^2 \\
&
\ge L\Big(\frac{2}{\pi}-\frac 1 L - \frac{2}{\pi L^2 q^2} - \frac{8s}{\pi L^2}\Big) >
L\Big(\frac{2}{\pi}- \frac{2}{\pi L^2 q^2} -\frac 4 L\Big).
\end{align*}
The separation condition \eqref{sep} in Theorem \ref{thm3} is derived from the positivity condition of the lower bound, i.e.,
$$\Big(\frac{2}{\pi}- \frac{2}{\pi L^2 q^2} -\frac 4 L\Big)>0.$$

We prove the upper bound in Theorem \ref{thm3} in two cases: $L$ is even or $L$ is odd.

\begin{description}
\item [Case 1: $L$ is even.]
First we substitute $L$ with $2L$ in \eqref{lemma61} and obtain
\beq
\label{lemma63}
\sum_{k= 0}^{2L} g(\frac{k}{2L}) \Big|(\Phi^{0\rightarrow 2L}  {\bc})_k\Big|^2 
\le 
2L\Big(\frac{2}{\pi} + \frac{1}{2L} + \frac{2}{4\pi L^2 q^2} + \frac{8s}{4\pi L^2}\Big)\|{\bc}\|_2^2.
\eeq

Let $D^{\frac L 2} = {\rm diag}(e^{-2\pi i \om_1 \frac L 2}, e^{-2\pi i \om_2 \frac L 2},\ldots,e^{-2\pi i \om_s \frac L 2})$ and $D^{-\frac L 2} = (D^{\frac L 2})^{-1}$.
On the one hand,
\begin{align*}
& \sum_{k= 0}^{2L} g(\frac{k}{2L}) \Big|(\Phi^{0\rightarrow 2L}  D^{-\frac L 2}{\bc})_k\Big|^2 
 \ge 
\sum_{k= L/2}^{3L/2} g(\frac{k}{2L}) \Big|(\Phi^{0\rightarrow 2L}  D^{-\frac L 2}{\bc})_k\Big|^2 
\ge
\sum_{k= L/2}^{3L/2} g(\frac 1 4) \Big|(\Phi^{0\rightarrow 2L}  D^{-\frac L 2}{\bc})_k\Big|^2 \\
& =
\frac{1}{\sqrt 2}\sum_{k= L/2}^{3L/2}  \Big|(\Phi^{0\rightarrow 2L}  D^{-\frac L 2}{\bc})_k\Big|^2 
=
\frac{1}{\sqrt 2} \sum_{k= 0}^{L}  \Big|(\Phi^{\frac L 2 \rightarrow \frac{3L}{2}} D^{-\frac L 2} {\bc})_k\Big|^2 
=
\frac{1}{\sqrt 2} \sum_{k= 0}^{L}  \Big|(\Phi^{ 0 \rightarrow L} D^{\frac L 2} D^{-\frac L 2} {\bc})_k\Big|^2 
\\
& =
\frac{1}{\sqrt 2} \sum_{k= 0}^{L}  \Big|(\Phi^{ 0 \rightarrow L}  {\bc})_k\Big|^2. 
\end{align*}

On the other hand, \eqref{lemma63} implies 
\begin{align*}
\sum_{k= 0}^{2L} g(\frac{k}{2L}) \Big|(\Phi^{0\rightarrow 2L}  D^{-\frac L 2}{\bc})_k\Big|^2 
& \le 
2L\Big(\frac{2}{\pi} + \frac{1}{2L} + \frac{2}{4\pi L^2 q^2} + \frac{8s}{4\pi L^2}\Big)\|D^{-\frac L 2}{\bc}\|_2^2
\\
& =L\Big(\frac{4}{\pi} + \frac{1}{L} + \frac{1}{\pi L^2 q^2} + \frac{4s}{\pi L^2}\Big)\|{\bc}\|_2^2.
\end{align*}

As a result
$$\|\Phi^L {\bc}\|_2^2
= \sum_{k= 0}^{L}  \Big|(\Phi^{ 0 \rightarrow L}  {\bc})_k\Big|^2
\le
L\Big(\frac{4\sqrt 2}{\pi} + \frac{\sqrt 2}{L} + \frac{\sqrt 2}{\pi L^2 q^2} + \frac{4\sqrt 2 s}{\pi L^2}\Big)\|{\bc}\|_2^2 
<
L\Big(\frac{4\sqrt 2}{\pi}  + \frac{\sqrt 2}{\pi L^2 q^2}+ \frac{3\sqrt 2}{L} \Big)\|{\bc}\|_2^2. 
$$
When $L$ is an even integer,
$$
\Big(\frac{2}{\pi}- \frac{2}{\pi L^2 q^2} -\frac 4 L\Big) \|\bc\|_2^2
\le
\frac{1}{L}\|\Phi^L {\bc}\|_2^2
\le
\Big(\frac{4\sqrt 2}{\pi}  + \frac{\sqrt 2}{\pi L^2 q^2}+ \frac{3\sqrt 2}{L} \Big)\|{\bc}\|_2^2. 
$$

\item [Case 2: $L$ is odd.]
\begin{align}
& \|\Phi^L {\bc}\|_2^2
= \sum_{k= 0}^{L}  \Big|(\Phi^{ 0 \rightarrow L}  {\bc})_k\Big|^2
\le
\sum_{k= 0}^{L+1}  \Big|(\Phi^{ 0 \rightarrow L+1}  {\bc})_k\Big|^2
\nonumber \\
& <
(L+1)\Big(\frac{4\sqrt 2}{\pi}  + \frac{\sqrt 2}{\pi (L+1)^2 q^2}+ \frac{3\sqrt 2}{L+1} \Big)\|{\bc}\|_2^2.
\label{lemma65} 
\end{align}

Eq. \eqref{lemma65} above follows from Case 1 as $L+1$ is an even integer. In summary, when $L$ is an odd integer,
$$
\left(\frac{2}{\pi}- \frac{2}{\pi L^2 q^2} -\frac 4 L\right) \|\bc\|_2^2
\le
\frac{1}{L}\|\Phi^L {\bc}\|_2^2
\le
\left(1+\frac 1 L\right)\left(\frac{4\sqrt 2}{\pi}  + \frac{\sqrt 2}{\pi (L+1)^2 q^2}+ \frac{3\sqrt 2}{L+1} \right)\|{\bc}\|_2^2. 
$$

\commentout{
Let $D^{\frac {L-1}{2}} = {\rm diag}(e^{-2\pi i \om_1 \frac {L-1}{2}}, e^{-2\pi i \om_2 \frac{L-1}{2}},\ldots,e^{-2\pi i \om_s \frac{L-1}{2}})$ and $D^{-\frac{L-1}{2}} = (D^{\frac{L-1}{2}})^{-1}$.

First we derive a lower bound of $g(\frac 1 4 - \frac{1}{4L})$.
\begin{align}
& g(\frac 1 4 - \frac{1}{4L}) = \cos(\frac \pi 4 +\frac{\pi}{4L}) = \cos(\frac \pi 4)\cos(\frac{\pi}{4L}) - \sin(\frac \pi 4)\sin(\frac{\pi}{4L}) \nonumber
\\
&\ge \frac{1}{\sqrt 2} \left( \cos(\frac{\pi}{4L}) - \sin(\frac{\pi}{4L}) \right) > \frac{1}{\sqrt 2} (1-\frac{1}{2L}-\frac{\pi}{4L}) >  
\frac{1}{\sqrt 2} (1-\frac{3}{2L}).
\label{lemma64}
\end{align}
The estimation is \eqref{lemma64} is based on $\cos(\frac{\pi}{4L}) > 1-\frac{1}{2L}$ and $\sin(\frac{\pi}{4L}) < \frac{\pi}{4L}$.

On the one hand,
\begin{align*}
& \sum_{k= 0}^{2L} g(\frac{k}{2L}) \Big|(\Phi^{0\rightarrow 2L}  D^{-\frac{L-1}{2}}{\bc})_k\Big|^2 
 \ge 
\sum_{k= (L-1)/2}^{(3L-1)/2} g(\frac{k}{2L}) \Big|(\Phi^{0\rightarrow 2L}  D^{-\frac{L-1}{2}}{\bc})_k\Big|^2 
\\
& \ge
\sum_{k= (L-1)/2}^{(3L-1)/2} g(\frac 1 4 - \frac{1}{4L}) \Big|(\Phi^{0\rightarrow 2L}  D^{-\frac{L-1}{2}}{\bc})_k\Big|^2 
>
\frac{1}{\sqrt 2}(1-\frac{3}{2L})\sum_{k= (L-1)/2}^{(3L-1)/2}  \Big|(\Phi^{0\rightarrow 2L}  D^{-\frac{L-1}{2}}{\bc})_k\Big|^2
\\ 
&
=
\frac{1}{\sqrt 2}(1-\frac{3}{2L})\sum_{k= 0}^{L}  \Big|(\Phi^{\frac{L-1}{2} \rightarrow \frac{3L-1}{2}} D^{-\frac{L-1}{2}} {\bc})_k\Big|^2 
=
\frac{1}{\sqrt 2}(1-\frac{3}{2L}) \sum_{k= 0}^{L}  \Big|(\Phi^{ 0 \rightarrow L} D^{\frac{L-1}{2}} D^{-\frac{L-1}{2}} {\bc})_k\Big|^2 
\\
& =
\frac{1}{\sqrt 2}(1-\frac{3}{2L}) \sum_{k= 0}^{L}  \Big|(\Phi^{ 0 \rightarrow L}  {\bc})_k\Big|^2. 
\end{align*}
}

\end{description}
\end{proof}

\end{appendices}
}

\end{document}